\documentclass[11pt,hidelinks,colorlinks=true]{article}

\pdfoutput=1
\usepackage{latexsym}
\usepackage{amsmath}
\usepackage{amssymb}
\usepackage{amsfonts,amsthm}
\usepackage[top=1in, bottom=1in, left=1in, right=1in]{geometry}

\usepackage{enumerate}

\usepackage[T1]{fontenc}
\usepackage{longtable}
\usepackage{tabu}
\usepackage{booktabs, multirow} 
\usepackage{soul}
\usepackage{changepage,threeparttable} 

\usepackage{tikz}
\usetikzlibrary{arrows,decorations.pathmorphing,decorations.shapes,backgrounds,fit}
\pgfkeys{tikz/cc1/.style={dotted} }
\pgfkeys{tikz/cc4/.style={dashed} }
\pgfkeys{tikz/cc3/.style={decorate,decoration=bumps} }
\pgfkeys{tikz/cc2/.style={decorate,decoration=snake} }
\pgfkeys{tikz/solid/.style={line width=2pt} }
\pgfkeys{tikz/solid2/.style={line width=1pt} }
\pgfkeys{tikz/nodostile1/.style={draw,line width=1pt}}
\pgfkeys{tikz/nodostile2/.style={draw,line width=2pt}}
\pgfkeys{tikz/arcostile1/.style={very thick} }
\pgfkeys{tikz/arcostile2/.style={red,very thick} }
\pgfkeys{tikz/arcostile3/.style={thick} }
\pgfkeys{tikz/arcostile4/.style={line width=2pt} }

\usepackage{array}
\newcolumntype{C}[1]{>{\centering\let\newline\\\arraybackslash\hspace{0pt}}m{#1}}

\newcommand{\comment}[1]{}

\usepackage{hyperref,xcolor}

\usepackage[lined, algoruled, linesnumbered]{algorithm2e}
\definecolor{winered}{rgb}{0.5,0,0}
\hypersetup
{
    pdfborder={0 0 0},
    colorlinks=true,
    linkcolor={winered},
    urlcolor={winered},
    filecolor={winered},
    citecolor={winered},
    linktoc=all,
}

\newtheorem{theorem}{Theorem}[section]
\newtheorem{lemma}[theorem]{Lemma}
\newtheorem{definition}[theorem]{Definition}
\newtheorem{corollary}[theorem]{Corollary}

\newtheorem{observation}[theorem]{Observation}

\newcommand{\ignore}[1]{}

\newcommand{\scon}{\leftrightarrow}

\newcommand{\predecessors}[2]{\mathit{Pred}_{#1}(#2)}
\newcommand{\successors}[2]{\mathit{Succ}_{#1}(#2)}

\newcommand{\dectree}{\mathcal{T}}

\newcommand{\ftreach}[4]{2\mathit{FTR}_{#1}(#2,#3,#4)}
\newcommand{\rftreach}[4]{2\mathit{FTR}^R_{#1}(#2,#3,#4)}

\newcommand{\ftsc}[3]{1\mathit{FTSC}(#1,#2,#3)}

\newcommand{\twoftsc}[4]{2\mathit{FTSC}(#1,#2,#3,#4)}

\newcommand{\note}[1]{\textbf{(*)}\marginpar {\tiny \raggedright{(*) #1}}}

\begin{document}

\title{\bf $2$-Fault-Tolerant Strong Connectivity Oracles\thanks{Department of Computer Science \& Engineering, University of Ioannina, Greece. E-mail: \texttt{\{loukas,ekosinas,dtsokaktsis\}@cse.uoi.gr}. Research supported by the Hellenic Foundation for Research and Innovation (H.F.R.I.) under the ``First Call for H.F.R.I. Research Projects to support Faculty members and Researchers and the procurement of high-cost research equipment grant'', Project FANTA (eFficient Algorithms for NeTwork Analysis), number HFRI-FM17-431.}}

\author{Loukas Georgiadis \and Evangelos Kosinas \and  Daniel Tsokaktsis}

\date{}

\maketitle

\begin{abstract}
We study the problem of efficiently answering strong connectivity queries under two vertex failures.
Given a directed graph $G$ with $n$ vertices, we provide a data structure with $O(nh)$ space and $O(h)$ query time, where $h$ is the height of a decomposition tree of $G$ into strongly connected subgraphs.
This immediately implies data structures with $O(n \log{n})$ space and $O(\log{n})$ query time for graphs of constant treewidth, and $O(n^{3/2})$ space and $O(\sqrt{n})$ query time for planar graphs.
For general directed graphs, we give a refined version of our data structure that achieves $O(n\sqrt{m})$ space and $O(\sqrt{m})$ query time, where $m$ is the number of edges of the graph.
We also provide some simple BFS-based heuristics that seem to work remarkably well in practice.
In the experimental part, we first evaluate various methods to construct a decomposition tree with small height $h$ in practice.
Then we provide efficient implementations of our data structures, and evaluate their empirical performance by conducting an extensive experimental study 
on graphs taken from real-world applications.
\end{abstract}

\section{Introduction}
\label{section:introduction}

Fundamental graph properties such as (strong) connectivity and reachability have been extensively studied for both undirected and directed graphs.
As real world networks are prone to failures, which can be unpredictable, the fault-tolerant (or sensitivity) model
has drawn the attention of several researchers in the recent past \cite{CC20:ICALP,choudhary:ICALP16,DBLP:journals/talg/Charalampopoulos22,Demetrescu:2008,DuanP09a,henzinger_et_al:LIPIcs:2017:8178,ParterP13,BrandSaranurak2019}.
Instead of allowing for an arbitrary sequence of updates, the fault-tolerant model only allows to apply batch
updates of small size to the original input data.
In this work we focus on constructing a data structure (oracle) that can answer strong connectivity queries between two vertices of a given directed graph (digraph) under any two vertex failures.

A strongly connected component (SCC) of a directed graph $G=(V,E)$ 
is a maximal subgraph of $G$ in which there is a directed path from each vertex to every other vertex.
The strongly connected components of $G$ partition the vertices of $G$ such that two vertices $x,y \in V$ are strongly connected (denoted by $x \leftrightarrow y$) if they belong to the same strongly connected component of $G$.
Computing the strongly connected components of a directed graph is one of the most fundamental graph problems that finds numerous applications in many diverse areas.
%
Thus, we would like to study efficient algorithms for determining the strong connectivity relation in the fault-tolerant model.
Towards such a direction, we wish to compute a small-size data structure for reporting efficiently whether two vertices are strongly connected under the possibility of vertex failures.
Usually, the task is to keep a data structure (oracle) that supports queries of the following form: for any two vertices $x,y$ and any set $F$ of $k$ vertices determine whether $x$ and $y$ are strongly connected in $G - F$.
More formally, we aim to construct an efficient fault-tolerant strong-connectivity oracle under possible
(bounded) failures.

\begin{definition}[FT-SC-O]
\label{def:ftsccoracle}
Given a graph $G=(V,E)$, a $k$-fault-tolerant strong-connectivity oracle ($k$-FT-SC-O) of $G$ is a data structure that, given a set of $k$ failed vertices\footnote{We note that a $k$-FT-SC-O construction for vertex failures can provide an oracle for intermixed vertex and edge failures by splitting every edge $e=(x,y)$ of the graph as $(x,e),(e,y)$, with the introduction of a new vertex $e$. Then, the removal of an edge $e$ from the original graph can be simulated by the removal of the vertex $e$ from the derived graph.} $f_1, \ldots, f_k\in V$ and two query vertices $x,y\in V$, it can determine (fast) whether $x$ and $y$ are strongly connected in $G-\{f_1, \ldots, f_k\}$.
\end{definition}

To measure the efficiency of an oracle, two main aspects are concerned: the size of the computed data structure and the running time for answering any requested query.
Ideally, we would aim for linear-size oracles with constant query time, but this seems out of reach for many problems~\cite{henzinger_et_al:LIPIcs:2017:8178}.
For instance, it is known that for a single vertex/edge failure (i.e., $k=1$) an oracle with $O(n)$ space and $O(1)$ query time is achievable~\cite{GIP20:SICOMP}.
However, for a larger number of failures (i.e., $k>1$) the situation changes considerably.
Even for $k=2$, straightforward approaches would lead to an oracle of $O(n^2)$ size with constant query time.

\subsection{Related work}

Maintaining the strongly connected components under vertex/edge updates 
has received much attention, both
in the 
dynamic setting, where the updates 
are permanent,
and in the fault-tolerant model, where the failures are part of the query.
%

\paragraph*{Fault-tolerant data structures.}
Baswana, Choudhary, and Roditty~\cite{baswana_et_al:ALGO:2019} presented a data structure of size $O(2^k n^2)$ that is computed in $O(2^k n^2 m)$ time, and outputs all strongly connected components in $O(2^k n \log^2{n})$ time under at most $k$ failures.
For $k=1$, Georgiadis, Italiano, and Parotsidis~\cite{GIP20:SICOMP} gave an $O(n)$-space single-fault strong connectivity oracle ($1$-FT-SC-O) that can report all strongly connected components in $O(n)$ time, and test strong connectivity for any two vertices in $O(1)$ time, under a single vertex/edge failure. 
A closely related problem is to be able to maintain reachability information under failures, either with respect to a fixed source vertex $s$ (single-source reachability) or with respect to a set of vertex pairs $\mathcal{P} \subseteq V \times V$ (pairwise reachability).
Choudhary~\cite{choudhary:ICALP16} presented a $2$-fault-tolerant single-source reachability oracle ($2$-FT-SSR-O) with $O(n)$ space that answers in $O(1)$ time whether a vertex $v$ is reachable from the source vertex $s$ in $G-\{f_1,f_2\}$, where $f_1,f_2$ are two failed vertices.
Later, Chakraborty, Chatterjee, and Choudhary~\cite{chakraborty_et_al:LIPIcs.ICALP.2022.35}, gave a $2$-fault-tolerant pairwise reachability oracle ($2$-FT-R-O) with $O(n \sqrt{|\mathcal{P}|})$ size that answers in $O(1)$ time whether a vertex $u$ reaches a vertex $v$ in $G-\{f_1,f_2\}$, for any pair $(u,v) \in \mathcal{P}$.
The above results imply $2$-FT-SC oracles of $O(n^2)$ size and $O(1)$ query time, either by storing a $1$-FT-SC-O~\cite{GIP20:SICOMP} of $G-v$ for all $v \in V$, or by storing a $2$-FT-SSR-O~\cite{choudhary:ICALP16} for all $v \in V$ as sources, or by setting $\mathcal{P} = V \times V$ in \cite{chakraborty_et_al:LIPIcs.ICALP.2022.35}.
Recently, van den Brand and Saranurak~\cite{BrandSaranurak2019} presented a Monte Carlo sensitive reachability oracle that preprocesses a digraph with $n$ vertices in $O(n^\omega)$ time and stores $O(n^2 \log{n})$ bits. Given a set of $k$ edge insertions/deletions and vertex deletions, the data structure is updated in $O(k^\omega)$ time and stores additional $O(k^2 \log{n})$ bits. Then,
given two query vertices $u$ and $v$, the oracle reports if there is directed path from $u$ to $v$ in $O(k^2)$ time.

For planar graphs, Italiano, Karczmarz, and Parotsidis~\cite{ftreachplanar:SODA21} showed how to construct a single-fault-tolerant all-pairs reachability oracle of $O(n \log n)$-space that answers in $O(\log{n})$ time whether a vertex $u$ reaches a vertex $v$ in $G-f$, where $f$ is a failed vertex or edge.
%
For more vertex failures, in weighted planar digraphs, Charalampopoulos, Mozes, and Tebeka~\cite{DBLP:journals/talg/Charalampopoulos22} provided an $O(n\log^2{n})$-time construction of an oracle of $O(n\log n)$ size, that can answer \emph{distance} queries for pairs of vertices in the presence of $k$ vertex failures in $O(\sqrt{kn}\log^2{n})$ time. This implies an oracle of $O(n\log n)$ size that can answer $2$-FT-SC queries in planar digraphs in $O(\sqrt{n}\log^2{n})$ time.

All the previous approaches yield data structures that require $\Omega(n^{2})$ space (for general digraphs), which is prohibitive for large networks. Thus, it is natural to explore the direction of trading-off space with query time.
Furthermore, within the fault-tolerant model, one may seek to compute a sparse subgraph $H$ of $G$ (called preserver) that enables to answer
(strong connectivity or reachability) queries under failures in $H$ instead of $G$, which can be done more efficiently since $H$ is sparse.
Chakraborty and Choudhary~\cite{CC20:ICALP} provided the first sub-quadratic (i.e., $O(n^{2-\epsilon})$-sized for $\epsilon >0$) subgraph that preserves the strongly connected components of $G$ under $k\geq 2$ edge failures, by showing the existence of a preserver of size $\widetilde{O}(k 2^k n^{2-1/k})$ that is computed by a polynomial (randomized) algorithm.
%
%

\paragraph*{Dynamic data structures.}
An alternative approach for answering queries under failures is via dynamic data structures.
%
To answer a query of the form: {\em ``Are $x$ and $y$ strongly connected in $G-\{f_1,f_2\}$?''}, for two failed vertices/edges $f_1$ and $f_2$, we can first delete $f_1$ and $f_2$, by updating the data structure, and then answer the query. To get ready to answer the next query we have to re-insert the deleted vertices/edges.
%
%
The main problem with this approach is that the update operation is often too time-consuming and leads to bad query time.
Furthermore, there is a conditional lower bound of $\Omega(m^{1-o(1)})$ update time for a single vertex (or edge) deletion for general digraphs \cite{AW14,HenzingerKNS15}.

Charalampopoulos, Mozes, and Tebeka ~\cite{DBLP:journals/talg/Charalampopoulos22} showed that a weighted planar digraph $G$ can be preprocessed in time $O(n\frac{\log^2{n}}{\log{\log{n}}})$, so that edge-weight updates, edge insertions not violating the planarity of $G$, edge deletions, vertex insertions and deletions, and distance queries can be performed in time $O(n^{2/3}\frac{\log^{5/3}{n}}{\log^{4/3}{\log{n}}})$ using $O(n)$ space.
Hence, this implies an $O(n)$-space data structure that can answer strong connectivity queries between two vertices under two edge failures in planar digraphs in $O(n^{2/3}\frac{\log^{5/3}{n}}{\log^{4/3}{\log{n}}})$ time.

\subsection{Our contribution}

We provide a general framework for computing dual fault-tolerant strong connectivity oracles based on a decomposition tree $\dectree$ of a digraph $G$ into strongly connected subgraphs.
Following \L{}\k{a}cki~\cite{scc-decomposition}, we refer to $\dectree$ as an \emph{SCC-tree of $G$}.
%
Specifically, the SCC-tree is obtained from $G$ by iteratively removing vertices in a specified order and computing the SCCs after every removal; the root of $\dectree$ corresponds to $G$, and the children of every node of $\dectree$ correspond to the SCCs of the parent node after removing a vertex.
We analyze our oracle with respect to the height $h$ of $\dectree$ 
(which depends on the chosen order of the removed vertices).
Then, by storing some auxiliary data structures~\cite{choudhary:ICALP16,GIP20:SICOMP} at each node of $\dectree$, we obtain the following result:

\begin{theorem}
\label{thm:dectree}
Let $G=(V,E)$ be a digraph on $n$ vertices and let $h$ be the height of an SCC-tree of $G$.
There is a polynomial-time algorithm that computes a $2$-FT-SC oracle for $G$ of size $O(nh)$ 
that answers strong connectivity queries  under two vertex failures in $O(h)$ time.
\end{theorem}

The construction time of the data structure in Theorem~\ref{thm:dectree} is $O(f(m,n)h)$, where $f(m,n)$ is the time to initialize a $2$-FT-SSR oracle of a graph with $n$ vertices and $m$ edges. In our experiments we implemented the oracle of Choudhary~\cite{choudhary:ICALP16}, which can be constructed in $f(m,n)=O(mn)$ time.

Despite the fact that there are graphs for which $h=\Omega(n)$, our experimental study reveals that the height of $\dectree$ is much smaller in practice.
To that end, we evaluate various methods to construct a decomposition tree with small height $h$ in practice.
We note that such SCC-trees are useful in various decremental connectivity algorithms (see, e.g., \cite{Chechik2016Decremental,decdom17,scc-decomposition}), so this experimental study may be of independent interest.

Theorem~\ref{thm:dectree} immediately implies the following results for special graph classes.

\begin{corollary}
\label{corollary:planar}
Let $G=(V,E)$ be a directed planar graph with $n$ vertices. 
There is a polynomial-time algorithm that computes a $2$-FT-SC oracle of $O(n\sqrt{n})$ size 
with $O(\sqrt{n})$ query time.
\end{corollary}

\begin{corollary}
\label{corollary:boundedtw}
Let $G=(V,E)$ be a directed graph, whose underlying undirected graph has treewidth 
bounded by a constant. 
There is a polynomial-time algorithm that computes a $2$-FT-SC oracle of $O(n\log{n})$ size 
with $O(\log{n})$ query time.
\end{corollary}

For general directed graphs, we also provide a refined version of our data structure that builds a  \emph{partial SCC-tree}, and achieves the following bounds.


\begin{theorem}
\label{thm:refined}
Let $G=(V,E)$ be a digraph on $n$ vertices and $m$ edges, and let $\Delta$ be an integer parameter in $\{1,\ldots,m\}$.
There is a polynomial-time algorithm\footnote{See Section~\ref{sec:improved} for more details on the construction time of the partial SCC-tree.} that computes a $2$-FT-SC oracle of $O(mn/\Delta)$ size with $O(m/\Delta + \Delta)$ query time.
\end{theorem}

Theorem~\ref{thm:refined} provides a trade-off between space and query time. To minimize the query time, we set $\Delta = \sqrt{m}$ which gives the following result.

\begin{corollary}
\label{corollary:refined}
Let $G=(V,E)$ be a digraph on $n$ vertices and $m$ edges.
There is a polynomial-time algorithm that computes a $2$-FT-SC oracle of $O(n\sqrt{m})$ size
with $O(\sqrt{m})$ query time.
\end{corollary}

Thus, when $m=o(n^2)$, the oracle of Corollary~\ref{corollary:refined} achieves $o(n^2)$ space and $o(n)$ query time. Furthermore, for sparse graphs, where $m=O(n)$, we have an oracle of $O(n^{3/2})$ space and $O(\sqrt{n})$ query time.

\ignore{
We note that \cite{chakraborty_et_al:LIPIcs.ICALP.2022.35} provides a lower bound of $\Omega(n\sqrt{|\mathcal{P}|})$ bits \note{Loukas: This is size in bits.} on the size of a $2$-FT-R oracle for answering reachability queries under two vertex failures for pairs of vertices in a set $\mathcal{P}\subseteq V\times V$ (irrespective of the query time).\note{Loukas: Here $P$ is a set of ``terminal'' vertex pairs, for which we wish to answer connectivity queries, and not a set of possible pairs of failures.} Notice that this is $\Omega(n^2)$ if $\mathcal{P}=V\times V$, i.e., if we want to be able to answer any reachability query under a pair of failures. Thus, Corollary~\ref{corollary:refined} implies that there is a sharp distinction between the optimal complexity of $2$-FT-R oracles and that of $2$-FT-SC oracles. }

For the experimental part of our work, we provide efficient implementations of our data structures and evaluate their empirical performance by conducting an extensive experimental study on graphs taken from real-world applications. Furthermore, we provide some simple BFS-based heuristics, that can be constructed very efficiently, they perform remarkably well in practice, and serve as the baseline for our experiments.
%

\section{Preliminaries}
\label{sec:preliminaries}

Let $G=(V,E)$ be a digraph.
For any subgraph $H$ of $G$, we denote by $V(H) \subseteq V$ the vertex set of $H$, and by $E(H) \subseteq E$ the edge set of $H$.
For $S\subseteq V,$ we denote by $G[S]$ the subgraph of $G$ induced by the vertices in $S$, and $G-S$ denotes the subgraph of $G$ that remains after removing the vertices from $S$. (Thus, we have $G-S=G[V\setminus S]$.)
Given a path $P$ in $G$ and two vertices $u,v\in V(P),$ we denote by $P[u,v]$ the subpath of $P$ starting from $u$ and ending at $v$.
If $P$ starts from $s$ and ends at $t$ we say that $P$ is an $s \to t$ path.
Two vertices $u,v\in V$ are \emph{strongly connected} in $G$, denoted by $u \scon v$, if there exist a $u \to v$ path and a $v \to u$ path in $G.$
The \emph{strongly connected components} (SCCs) of $G$ are its maximal strongly connected subgraphs.
Thus, two vertices $u,v \in V$  are strongly connected if and only if they belong to the same strongly connected component of $G$.
The \emph{size} of a strongly connected component is given by the number of its edges.
It is easy to see that the SCCs of $G$ partition its vertex set.
The \emph{reverse digraph} of $G$, denoted by $G^R$, is obtained from $G$ by reversing the direction of all edges.
The predecessors (resp., successors) of a vertex $v$ in $G$, denoted by $\predecessors{G}{v}$ (resp., $\successors{G}{v}$), is the set of vertices that reach $v$ (resp., are reachable from $v$) in $G$.

A vertex of $G$ is a \emph{strong articulation point} (SAP) if its removal increases the number of strongly connected components.
A strongly connected digraph $G$ is \emph{$2$-vertex-connected} if it has at least three vertices and no strong articulation points.
Similarly, two vertices $f_1,f_2 \in V$ form a \emph{separation pair} if their removal increases the number of strongly connected components.
A strongly connected digraph $G$ is \emph{$3$-vertex-connected} if it has at least four vertices and no separation pairs.
Note that a SAP of $G$ forms a separation pair with any other vertex, and so we make the following distinction.
We say that a separation pair $\{f_1,f_2\}$ is \emph{proper} if $f_2$ is a SAP of $G-f_1$ and $f_1$ is a SAP of $G-f_2$.

\paragraph*{Auxiliary data structures.}
Consider a digraph $G$ with $n$ vertices, and let $s$ be a designated start vertex.
Our oracles make use of the following auxiliary data structures for $G$. 

\begin{description}

\item[$2$-FT-SSR-O.] From \cite{choudhary:ICALP16}, there is an $O(mn)$-time algorithm that computes a dual-fault-tolerant single-source reachability oracle ($2$-FT-SSR-O) of $O(n)$ size that answers in $O(1)$ time reachability queries of the form ``{\em is vertex $v$ reachable from $s$ in $G-\{f_1,f_2\}$?}'', where the vertex $v\in V(G)$ and the failed vertices $f_1,f_2\in V(G)$ are part of the query.
    We denote by $\ftreach{s}{v}{f_1}{f_2}$ the answer to such a query.
    Moreover, we use a similar $2$-FT-SSR oracle for $G^R$, i.e., an oracle of $O(n)$ size that answers in $O(1)$ time reachability queries of the form ``{\em is vertex $s$ reachable from $v$ in $G-\{f_1,f_2\}$?}\,''.
    We denote by $\rftreach{s}{v}{f_1}{f_2}$ the answer to such a query.
    %

\item[$1$-FT-SC-O.] From \cite{GIP20:SICOMP}, there is a linear-time algorithm that computes a single-fault-tolerant strong-connectivity oracle ($1$-FT-SC-O) of $O(n)$ size that answers in $O(1)$ time queries of the form ``{\em are vertices $x$ and $y$ strongly connected in $G-f$?}\,'', where the vertices $x,y\in V(G)$ and the failed vertex $f \in V(G)$ are part of the query.
    We denote by $\ftsc{x}{y}{f}$ the answer to such a query.
\end{description}

We state a simple fact that will be useful in our query algorithms.

\begin{observation}
\label{obervation:sc}
For any two vertices $x,y \in V(G)-s$, we have $x \scon y$ in $G-\{f_1,f_2\}$ only if $\ftreach{s}{x}{f_1}{f_2}= \ftreach{s}{y}{f_1}{f_2}$ and $\rftreach{s}{x}{f_1}{f_2}=\rftreach{s}{y}{f_1}{f_2}$.
\end{observation}
\section{Decomposition tree}
\label{sec:dectree}

We construct an SCC-tree $\dectree$ of a strongly connected digraph $G=(V,E)$ based on the idea introduced by \L{}\k{a}cki~\cite{scc-decomposition}.
(If the input digraph $G$ is not strongly connected, then we construct a separate SCC-tree for each strongly connected component.)
Each node $N(t)$ of $\dectree$ corresponds to a vertex $t$ of $G$, referred to as the \emph{split vertex} of $N(t)$.
Also, a node $N(t)$ of $\dectree$ is associated with a subset of vertices $S_{t}\subseteq V$ that contains $t$.
%
%
An SCC-tree of $G$ is constructed recursively as follows:
\begin{itemize}
    \item We choose a split vertex $r$ of $G$, and let $N(r)$ be the root of $\dectree$. We associate $N(r)$ with $S_{r}=V(G)$.
    \item For a node $N(t) \in \dectree$ such that $|S_t| \ge 2$, let $H_{1}, \ldots, H_{k}$ be the SCCs of $G[S_t]-t$. For every $i\in \{1,\ldots,k\}$ we choose a split vertex $t_{i}\in V(H_{i})$ and make the corresponding node $N(t_i)$ a child of $N(t)$. We set $S_{t_i} = V(H_i)$, and recursively compute an SCC-tree for $G[S_{t_i}]$ rooted at $N(t_i)$.
\end{itemize}

Figure~\ref{fig:scctree-main} gives an example of an SCC-tree.
For every vertex $x\in V$, we define $P_{x}$ to be the set of nodes $N(t)$ in $\dectree$ such that $x \in S_{t}$. Thus, $P_x$ consists of the nodes on the tree-path from $N(r)$ to $N(x)$.
We also let $V_{x}$ be the set of split vertices that correspond to $P_{x}$, i.e., $V_x = \{ t \in V : N(t) \in P_{x}\}$.

\ignore{
\begin{figure*}
 	\centering
  \includegraphics[clip=false, width=0.8\textwidth]{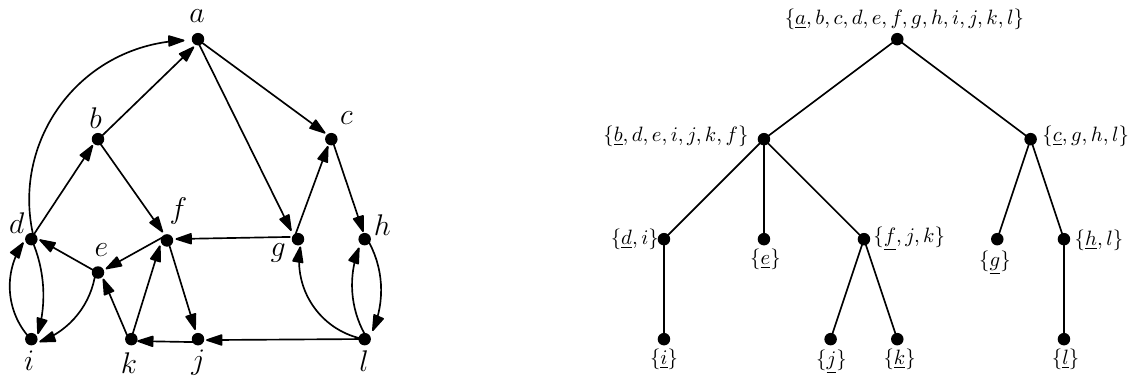}
 	\caption{A strongly connected digraph $G$ (left) and a SCC-tree $\dectree$ of $G$ (right). Every node of $\dectree$ is associated with a subset of $V(G)$ and the underlined vertex is the corresponding split vertex. For example, the left child of the root is $N(b)$ and $S_{b}=\{b,d,e,i,j,k,f\}$. Also note that $P_{l} = \langle N(a),N(c),N(h), N(l) \rangle$.}\label{fig:scctree}
 \end{figure*}
}

Observe that the number of nodes of $\dectree$ is exactly $|V|$ and every $v\in V$ appears exactly once as a split vertex in a node $N(v) \in \dectree$.
Thus there is a one-to-one correspondence between the vertices of $G$ and the nodes of $\dectree$.
For every node $N(t) \in \dectree$, we define $G_{t}$ to be the strongly connected subgraph of $G$ induced by $S_{t}$, i.e., $G_{t}=G[S_{t}]$.
By construction, it follows that for every $x\in V$ the nodes of $P_{x}$ form a path (starting from the root) in $\dectree$.
Thus we can think of $P_{x}$ as an ordered set (where its elements are ordered from the root to the node $N(x)$ of $\dectree$) and we denote by $P_{x}(i)$ its $i$th element (if $|P_{x}|<i$ then $P_{x}(i)=\mathit{null}$).
For $x,y\in V$, we define their {\em nearest common ancestor}, $\mathit{nca}(x,y)$, in $\dectree$ to be the last common element of $P_{x}$ and $P_{y}$.
%

%
At each node $N(t) \in \dectree$, we store the auxiliary data structures of Section~\ref{sec:preliminaries}, which we use to answer a query, as we describe next.
%

\begin{figure*}[h!]
 	\centering
  \includegraphics[clip=false, width=\textwidth]{scc_tree.pdf}
 	\caption{A strongly connected digraph $G$ (left) and an SCC-tree $\dectree$ of $G$ (right). Every node of $\dectree$ is associated with a subset of $V(G)$ and the underlined vertex is the corresponding split vertex. For example, the left child of the root is $N(b)$ and $S_{b}=\{b,d,e,i,j,k,f\}$. Also note that $P_{l} = \langle N(a),N(c),N(h), N(l) \rangle$.}\label{fig:scctree-main}
 \end{figure*}

\paragraph*{Answering a query.} 
Now we describe an algorithm which, given an SCC-tree $\dectree$ of a strongly connected graph $G=(V,E)$, two query vertices $x,y\in V$ and two failed vertices $f_1,f_2\in V,$ answers the query $\twoftsc{x}{y}{f_1}{f_2}$ that asks whether $x$ and $y$ are strongly connected in $G-\{f_1,f_2\}$.
The algorithm begins at the root $N(r)$ of $\dectree$ and descends the path $P_{\mathit{nca}(x,y)} = P_x \cap P_y$.
When we visit a node $N(t)$ we perform the following steps:

\begin{enumerate}
\item If $t$ is a failed vertex, say $t=f_1$, then we check if $N(t)=\mathit{nca}(x,y)$ (i.e., if $N(x)$ and $N(y)$ are descendants of different children of $N(t)$). If this is the case then we return \textsc{false}. Otherwise, we return the result of the query $\ftsc{x}{y}{f_2}$ for $G_w$, where $N(w)$ is the child of $N(t)$ that is an ancestor of both $N(x)$ and $N(y)$. ($N(w)$ is the next node on $P_{\mathit{nca}(x,y)}$.)

\item If $t$ is not a failed vertex, then we test the condition (C): $( \ftreach{t}{x}{f_1}{f_2} \not= \ftreach{t}{y}{f_1}{f_2} ) \vee ( \rftreach{t}{x}{f_1}{f_2} \not= \rftreach{t}{y}{f_1}{f_2} )$ in $G_t$. If it is true, then we return \textsc{false}.

\item If (C) is false and both $\ftreach{t}{x}{f_1}{f_2}$ and $\rftreach{t}{x}{f_1}{f_2}$ are true in $G_t$, then we return \textsc{true}. Otherwise, we proceed to the next node on $P_{\mathit{nca}(x,y)}$.
\end{enumerate}





\begin{lemma}
The query algorithm is correct.
\end{lemma}
\begin{proof}

Consider a query $\twoftsc{x}{y}{f_1}{f_2}$ that asks if $x$ and $y$ are strongly connected in $G-\{f_1,f_2\}$.
We first argue that if our procedure returns \textsc{true}, then $x$ and $y$ are strongly connected in $G-\{f_1,f_2\}$.
This happens in one of the following cases:
\begin{itemize}
\item One of the failing vertices, say $f_1$, is the split vertex $t$ in the currently visited node $N(t)$ of $\dectree$. Let $N(w)$ be the child of $N(t)$ containing $x$ and $y$. (This node exists because otherwise, the query algorithm would return \textsc{false}.) We have two cases:

\begin{itemize}
\item  The other failing vertex, $f_2 \in S_w$. Then, we return the answer $\ftsc{x}{y}{f_2}$ for $G_w$, which is \textsc{true} if and only if $x \scon y$ in $G_t - \{f_1,f_2\}$. Hence, $x$ and $y$ are strongly connected in $G-\{f_1,f_2\}$.

\item Vertex $f_2 \notin S_w$. Then, we return \textsc{true} since $x$ and $y$ are in the same SCC of $G_t-t$, induced by the vertices of $S_w$, which does not contain any failed vertices.
\end{itemize}

\item The split vertex $t$ of $N(t)$ is not a failed vertex. The query algorithm returns \textsc{true} when we have $\ftreach{t}{x}{f_1}{f_2} = \ftreach{t}{y}{f_1}{f_2}$ and $\rftreach{t}{x}{f_1}{f_2} = \rftreach{t}{y}{f_1}{f_2}$, and also both $\ftreach{t}{x}{f_1}{f_2}$ and $\rftreach{t}{x}{f_1}{f_2}$ are true in $G_t$. Then, we have that both $x$ and $y$ are strongly connected with $t$ in $G_{t}-\{f_1,f_2\}$. Thus, $x$ and $y$ are strongly connected in $G-\{f_1,f_2\}$.
\end{itemize}

For the opposite direction, suppose that $x$ and $y$ are strongly connected in $G-\{f_1,f_2\}$. Let $C$ be the SCC of $G-\{f_1,f_2\}$ that contains both $x$ and $y$. We 
argue that our procedure will return a positive answer.
The vertices of $C$ are contained in $S_{r}$, where $r$ is the split vertex of the root node $N(r)$ of the SCC-tree $\dectree$ (we recall that $S_{r}=V(G)$), meaning that $P_{x}(1)=P_{y}(1)$. Let $k$ be the positive integer such that $\mathit{nca}(x,y)=P_{x}(k)=P_{y}(k)$.
Assume that the query algorithm has visited the first $i \le k$ nodes of $P_{\mathit{nca}(x,y)}$ without returning a positive answer.
Then, $C$ does not contain any of the split vertices of the first $i$ nodes of $P_{\mathit{nca}(x,y)}$.
Hence, if no positive answer has been returned until step $k-1$, then for every $i\in\{1,\ldots k-1\}$, $C$ was entirely contained in all sets $S_{t_{i}},$ where $t_{i}$ is the split vertex of $N(t_i)=P_{x}(i)=P_{y}(i)$. By the definition of $k$, the next node $t=t_k$ considered by the algorithm has at least two distinct children that contain $x$ and $y$, respectively. This implies that $t$ is a vertex of $C$. Then, we have $\ftreach{t}{x}{f_1}{f_2} = \ftreach{t}{y}{f_1}{f_2}$ and $\rftreach{t}{x}{f_1}{f_2} = \rftreach{t}{y}{f_1}{f_2}$, and also both $\ftreach{t}{x}{f_1}{f_2}$ and $\rftreach{t}{x}{f_1}{f_2}$ are true in $G_t$. So the algorithm returns \textsc{true}.
\end{proof}

\paragraph*{Space and running time.} Regarding the running time of the query, we observe that the query algorithm makes at most $O(h)$ queries to the auxiliary data structures, where $h$ is the height of the SCC-tree $\dectree$. As each auxiliary structure has constant query time,  the oracle provides the answer in $O(h)$ time.
Regarding space, note that at each node $N(t) \in \dectree$, we store the auxiliary data structures of Section~\ref{sec:preliminaries}, which require $O(|S_t|)$ space.
Hence, our oracle occupies $\sum_{t \in V(G)} O(|S_t|) = O(nh)$ space.
This concludes the proof of Theorem~\ref{thm:dectree}.




\paragraph*{Implementation details.}
The oracle of Choudhary~\cite{choudhary:ICALP16} computes detour paths with respect to two divergent spanning trees~\cite{DomCert:TALG} $T_1$ and $T_2$ of $G$.
The spanning trees $T_1$ and $T_2$ are rooted at $s$ and have the property that for any vertex $v \not=s$, the only common vertices on the two tree $s$-$v$ paths  are the ones included on all $s$-$v$ paths in $G$.
Moreover, $T_1$ and $T_2$ can be computed in $O(m)$ time~\cite{DomCert:TALG}.
To answer a query $\ftreach{s}{x}{f_1}{f_2}$, \cite{choudhary:ICALP16} uses a data structure for reporting minima on tree paths~\cite{CartesianTrees:Demaine}.
Specifically, Demaine et al.~\cite{CartesianTrees:Demaine} show that a tree $T$ on $n$ vertices and edge weights can be preprocessed in $O(n \log{n})$ time to build a data structure of $O(n)$ size so that given any $u, v \in T$, the edge of smallest weight on the tree path from $u$ to $v$ can be reported in $O(1)$ time.

\subsection{Special graph classes}
\label{sec:special_classes}

The space and query time of the oracle depends on the value of the parameter $h$ (the height of the SCC-tree), which can be $O(n)$.
 For restricted graph classes, we can choose the split vertices in a way that guarantees better bounds for $h$. Such classes are planar graphs and bounded treewidth graphs.
 
\begin{definition}
\label{def:vertex separator}
A {\em vertex separator} of an undirected graph $G=(V,E)$ is a subset of vertices, whose removal decomposes the graph into components of size at most $\alpha|V|,$ for some constant $0<\alpha<1.$ A family of graphs $\cal{F}$ is called $f(n)$-separable if
\begin{itemize}
\item for every $F\in \cal{F},$ and every subgraph $H\subseteq F,$ $H\in \cal{F},$
\item for every $F\in \cal{F},$ such that $n=|V(F)|,$ $F$ has a vertex separator of size $f(n).$
\end{itemize}
\end{definition}

\begin{lemma}[\cite{scc-decomposition}]
\label{thm:decomp}
 Let $G=(V,E)$ be a directed strongly connected graph, such that $G\in \cal{F}$ is $Cn^{s}$-separable ($s\geq 0$). Moreover, assume that the separators for every graph $\cal{F}$ can be found in linear time. Then, we can build an SCC-decomposition tree for $G$ of height $O(h(n))$ in $O(|E|h(n))$ time, where $h(n)=O(n^{s})$ for $s>0$ and $h(n)=O(\log n)$ for $s=0.$
\end{lemma}

We next show the applicability of Lemma~\ref{thm:decomp} by providing efficient $2$-FT-SC oracles on well-known graph classes with structured underlying properties.

\paragraph*{Planar graphs.}
Here we assume that the underlying undirected graph is planar. The following size of separators in planar graphs is well-known.

\begin{theorem}[\cite{planar-separators:LT}] Planar graphs are $\sqrt{8n}$-separable and the separators can be found in linear time.
\end{theorem}

Combined with Lemma~\ref{thm:decomp}, the previous result when $G$ is planar yields Corollary~\ref{corollary:planar}.

\paragraph*{Graphs of bounded treewidth.}
Here we consider graphs for which their underlying undirected graph has bounded treewidth.
We will not insist on the formal definition of treewidth, as we only need the following fact.

\begin{theorem}[Reed \cite{Reed92}]
\label{thm:treewidth}
Graphs of treewidth at most $k$ are $k$-separable. Assuming that $k$ is constant, the separators can be found in linear time.
\end{theorem}

It is known that graphs with constant treewidth have $O(n)$ edges (see Reed \cite{Reed92}). This fact combined with Theorem~\ref{thm:treewidth} and Lemma~\ref{thm:decomp} implies that when $G=(V,E)$ is a directed graph, whose treewidth (of its underlying undirected graph) is bounded by a constant $k$, then we can build an SCC-decomposition tree for $G$ of height $O(\log n)$ in $O(n\log n)$ time. Thus, we obtain Corollary~\ref{corollary:boundedtw}.

\section{Improved data structure for general graphs}
\label{sec:improved}

In this section, we present an improved data structure for general graphs. Our data structure uses $O(n\sqrt{m})$ space and answers strong connectivity queries in $O(\sqrt{m})$ time.
The main idea of the improved oracle is that when we build the SCC-tree $\dectree$, we can stop the decomposition of a subgraph $G_t$ early if some appropriate conditions are satisfied (e.g., if $G_t$ is $3$-vertex connected). We refer to such a decomposition tree $\dectree$ of $G$ as a \emph{partial SCC-tree}.
Let $\Delta$ be an integer parameter in $[1,m]$. We say that a subgraph $G'$ of $G$ is ``large'' if it contains at least $\Delta+1$ edges, and we call it ``small'' otherwise.
%
Also, we say that a strongly connected graph $G$ is ``$\Delta$-good'' if it has the following property: For every separation pair $\{f_1,f_2\}$ of $G$, the graph $G'= G - \{f_1,f_2\}$ satisfies the following: (1) it has at most one large SCC $C$, and (2) for every vertex $v$ of $G'$ not belonging to $C$ it holds that either $G[\predecessors{G'}{v}\cup\{f_1,f_2\}]$ or $G[\successors{G'}{v}\cup\{f_1,f_2\}]$ contains at most $\Delta$ edges.
%

The following lemma shows that all 2FT-SC queries on a $\Delta$-good graph can be answered in $O(\Delta)$ time. More precisely, we can derive the answer to the query after accessing at most $4\cdot\Delta+4$ edges. This is done by performing four local searches with threshold $\Delta+1$ (starting from every one of the two query vertices, in both $G$ and $G^R$), without going beyond the failed vertices.
Since the graph is $\Delta$-good, the idea is that if a vertex $x$ can reach at least $\Delta+1$ edges and can be reached by at least $\Delta+1$ edges in $G-\{f_1,f_2\}$, then we can immediately infer that $x$ lies in the large SCC $C$ of $G-\{f_1,f_2\}$. 

\begin{lemma}
\label{lemma:Delta-good}
Let $G$ be a $\Delta$-good graph. Then, any $2$-fault strong connectivity query can be answered in $O(\Delta)$ time.
\end{lemma}
\begin{proof}
Consider a query that asks if vertices $x$ and $y$ are strongly connected in $G'= G - \{f_1,f_2\}$.
Note that $x \leftrightarrow y$ in $G'$ if and only if $y \in \successors{G'}{x}$ and $x \in \successors{G'}{y}$, which implies
$\successors{G'}{x} = \successors{G'}{y}$ (equivalently,  $\predecessors{G'}{x} = \predecessors{G'}{y}$).
To answer the query, for $z \in \{x, y\}$, we run simultaneously searches from and to $z$ (by executing a BFS or DFS from $z$ in $G'$ and $(G')^R$, respectively) in order to discover the sets $\successors{G'}{z}$ and $\predecessors{G'}{z}$.
(To perform a search in $G'$, we execute a search in $G$ but without expanding the search from $f_1$ or $f_2$ if we happen to meet them.)
We stop such a search early, as soon as the number of traversed edges reaches $\Delta+1$.
If this happens both during the search for $\predecessors{G'}{z}$ and for $\successors{G'}{z}$, then we conclude that $z \in C$.
Thus, if the four searches for $\predecessors{G'}{z}$ and $\successors{G'}{z}$, for $z \in \{x, y\}$, are stopped early, we know that both $x$ and $y$ belong to $C$ and so they are strongly connected.

Now, suppose that the search for $\successors{G'}{x}$ traversed at most $\Delta$ edges. Then, $x \leftrightarrow y$ in $G'$ only if the search for $\successors{G'}{y}$ also traversed at most $\Delta$ edges. Hence, if the search for $\successors{G'}{y}$ stopped early, we know that $x$ and $y$ are not strongly connected in $G'$. Otherwise, we just need to check if $y \in \successors{G'}{x}$ and $x \in \successors{G'}{y}$.
The case where the search for $\predecessors{G'}{x}$ traversed at most $\Delta$ edges is analogous.

So, in every case, we can test if $x$ and $y$ are strongly connected in $G'$ in $O(\Delta)$ time.
\end{proof}

We call a separation pair $\{f_1,f_2\}$ of $G$ ``good'' if every SCC of $G - \{f_1,f_2\}$ is small.
Now, to build a partial SCC-tree $\dectree$, we distinguish the following cases.

\paragraph*{Case 1: $G$ is small.} We stop the decomposition here, because all queries on $G$ can be answered in $O(\Delta)$ time.

\paragraph*{Case 2: $G$ is $3$-vertex-connected.} Then, $G$ does not contain any separation pairs, so all queries are answered (in the affirmative) in $O(1)$ time.

\paragraph*{Case 3: $G$ contains a good separation pair $\{f_1,f_2\}$.} Here we choose $\{f_1,f_2\}$ as the next two split vertices of $G$, because then the decomposition will terminate in two more levels (since the grandchildren of $G$ in the decomposition tree will correspond to small graphs).

\paragraph*{Case 4: $G$ is $\Delta$-good.} Then we stop the decomposition here, because all queries can be answered in $O(\Delta)$ time by 
Lemma~\ref{lemma:Delta-good}.

\paragraph*{Case 5: None of the above applies.} For this case, we prove that $G$ has the following property:

\begin{lemma}
\label{lemma:case5}
There is at least one separation pair $\{f_1,f_2\}$ of $G$ such that every SCC of $G'=G - \{f_1,f_2\}$ is either small, or contains fewer than $m-\Delta$ edges.
\end{lemma}
\begin{proof}
Since $G$ is not $3$-vertex-connected, there is at least one separation pair.
Also, since $G$ is not $\Delta$-good, there exists at least one separation pair $\{f_1,f_2\}$ with the property that either: (a) $G - \{f_1,f_2\}$ contains more than one large SCC, or (b) $G - \{f_1,f_2\}$ contains only one large SCC $C$, and for at least one vertex $v$ of $G'$ that does not belong to $C$, we have
that both $G[\predecessors{G'}{v} \cup \{f_1,f_2\}]$ and $G[\successors{G'}{v} \cup \{f_1,f_2\}]$ contain at least $\Delta$ edges.

If (a) is true, then the Lemma holds since all SCCs of $G'$ contain fewer than $m-\Delta$ edges.
Now suppose that (b) is true. Since $C$ is a SCC of $G'$, we either have $\successors{G'}{v} \cap C = \emptyset$ or $C \subset \successors{G'}{v}$.
If $\successors{G'}{v}$ does not contain $C$, then $C$ has fewer than $m-\Delta$ edges, and all the other strongly connected components of $G'$ have size at most $\Delta$.
Hence, the Lemma holds.
Otherwise, $C \subset \successors{G'}{v}$, and since $v \not\in C$, we have $\predecessors{G'}{v} \cap C = \emptyset$.
Since $G[\predecessors{G'}{v} \cup \{f_1,f_2\}]$ contains at least $\Delta$ edges, $C$ has fewer than $m-\Delta$ edges. Hence, the Lemma holds in this case as well.
\end{proof}

Obviously, only the last case may lead to repeated decompositions of $G$, but due to Lemma~\ref{lemma:case5} this occurs at most $m/\Delta$ times.
Thus, the decomposition tree has height $O(m/\Delta)$, and so it requires $O(mn/\Delta)$ space. Moreover, the queries can be answered in $O(m/\Delta + \Delta)$ time.
This proves Theorem~\ref{thm:refined}.
The running time is minimized for $\Delta = \sqrt{m}$, which gives Corollary~\ref{corollary:refined}.

The initialization time of this data structure is $O(m/\Delta\cdot(f(n,m)+g(n,m,\Delta)))$, where $f(n,m)$ is the time to initialize a $2$-FT-SSR oracle on a graph with $n$ vertices and $m$ edges, and $g(n,m,\Delta)$ is the time that is needed in order to check whether a graph with $n$ vertices and $m$ edges is $\Delta$-good.

Although we may always set $\Delta=\sqrt{m}$ in order to minimize both $O(\Delta)$ and $O(m/\Delta)$ (the height of the decomposition tree), in practice we may have that a small enough $\Delta$ may be able to provide a partial SCC-tree with height $O(\Delta)$. For example, this is definitely the case when $G$ itself is $\Delta$-good. 
Otherwise, it may be that after deleting a few pairs of vertices, we arrive at subgraphs that are $\Delta$-good.

Table~\ref{tab:ImprovedTree} shows some examples of real graphs where we have computed a value for $\Delta$ such that the partial SCC-tree has height at most $\Delta$. Thus, we get data structures for those graphs that can answer 2FT-SC queries in $O(\Delta)$ time. We arrived at those values for $\Delta$ by essentially performing binary search, in order to find a $\Delta$ that is as small as possible and such that either the graph is $\Delta$-good, or it has a partial SCC-tree decomposition with height at most $\Delta$.

\subsection{Testing if a graph is $\Delta$-good}

The computationally demanding part here is to determine whether a graph is $\Delta$-good, for a specific $\Delta$. The straightforward method that is implied by the definition takes $O(n^2(m+n\Delta))$ time. (I.e., this simply checks the SCCs after removing every pair of vertices, and it performs local searches with threshold $\Delta+1$ starting from every vertex.) Instead, we use a method that takes $O(nm+\Sigma_{v\in V(G)}\mathit{SCC}_v(G)\Delta)$ time, where $\mathit{SSC}_v(G)$ denotes the total number of the SCCs of all graphs of the form $G-\{v,u\}$, for every vertex $u\in V(G)\setminus v$. This is more efficient than the straightforward method, and it seems to work much better than the stated bound, because we find that $\mathit{SSC}_v(G)$, although it can theoretically be as large as $\Omega(n^2)$, in practice it is usually close to $n$, for every $v\in V(G)$.

The idea is to check the SCCs after the removal of every vertex $v\in V(G)$ (this explains the $O(nm)$ part). If $G- v$ has at least three large SCCs, then we can immediately conclude that $G$ is not $\Delta$-good. Otherwise, we distinguish three cases, depending of whether $G- v$ has $0$, $1$ or $2$ large SCC(s). In the first case, we can terminate the computation, because all SCCs of $G- v$ are small. In the other two cases, we essentially rely on the work~\cite{GIP20:SICOMP}, with which we can compute in $O(\mathit{SCC}_v(G))$ time all the strongly connected components of $G-\{v,u\}$, for every vertex $u\in V(G)\setminus v$, by exploiting information from the dominator trees and the loop nesting forests of the SCCs of $G- v$. For every such small component, it is sufficient to select a representative vertex $x$, and perform the two local searches from $x$ in $G$ and $G^R$ with threshold $\Delta+1$ (and blocking vertices $v$ and $u$), in order to determine whether $x$ either reaches at most $\Delta$ edges, or is reached by at most $\Delta$ edges.

Still, our result on the partial SCC-tree decomposition (Corollary~\ref{corollary:refined}) is mainly of theoretical interest, since the procedure for determining whether a graph is $\Delta$-good becomes very slow even in moderately large graphs. Thus, in Sections~\ref{sec:BFS} and \ref{sec:empirical} we suggest much more efficient heuristics, that work remarkably well in practice.

\section{BFS-based oracles}
\label{sec:BFS}

The straightforward way to determine whether two vertices $x$ and $y$ remain strongly connected after the removal of $f_1,f_2$, is to perform a graph traversal (e.g., BFS) in order to check whether $x$ reaches $y$ in $G-\{f_1,f_2\}$, and reversely. We call this algorithm $\mathtt{simpleBFS}$. We measure the work done by $\mathtt{simpleBFS}$ by counting the number of edges that we had to access in order to get the answer.

A well-established improvement over the simple BFS is the \emph{bidirectional} BFS \cite{hanauer_et_al:LIPIcs:2020:12088}. This works by alternating the search from $x$ to $y$ in $G$ with a search from $y$ to $x$ in $G^R$ (in order to determine whether $x$ reaches $y$). If either traversal reaches a vertex that was discovered by the other, then both terminate, and the answer is positive. If either search gets stuck and is unable to make progress, we conclude that the answer is negative. We implemented the variant where the searches alternate immediately after discovering a new edge, and we call our implementation $\mathtt{biBFS}$. Thus, $\mathtt{biBFS}$ works as follows. First, we perform a bidirectional BFS in order to determine if $x$ reaches $y$; if that is the case, then we perform a second bidirectional BFS in order to determine if $y$ reaches $x$. Again, we measure the work done by $\mathtt{biBFS}$ by counting the number of edges that we had to access in order to get the answer. As expected, $\mathtt{biBFS}$ has to do much less work on average than $\mathtt{simpleBFS}$, and thus we use $\mathtt{biBFS}$ as the baseline.

One of our most important contributions is a heuristic that we call \emph{seeded} BFS. This precomputes some data structures on a few (random) vertices that we call seeds, so that we can either use them directly before initiating the BFS, or if we meet them during the BFS.\footnote{In the literature, the vertices that support such functionality are commonly called supportive vertices, or landmarks \cite{hanauer_et_al:LIPIcs:2020:12088, Goldberg:Astar}.} Specifically, we use every seed $r$ in order to expand a BFS tree $\mathit{BFS}_r$ of $G$ with root $r$, and we maintain the preordering of all vertices w.r.t. $\mathit{BFS}_r$, as well as the number of descendants $\mathit{ND}_r(v)$ on $\mathit{BFS}_r$ for every vertex $v\in V(G)$. This information can be computed in linear time (e.g., after a DFS on $\mathit{BFS}_r$), and it can be used in order to answer ancestry queries w.r.t. $\mathit{BFS}_r$ in $O(1)$ time \cite{dfs:t}. We do the same on $G^R$ with the same seed vertices. 

Now, in order to answer an SC query for $x$ and $y$ in $G-\{f_1,f_2\}$, we first perform a bidirectional BFS from $x$ to $y$, with the following twist: if we meet a seed $r$, then we check whether either of $f_1,f_2$ is an ancestor of $y$ on $\mathit{BFS}_r$. If neither of $f_1,f_2$ is an ancestor of $y$, then we can immediately conclude that $x$ reaches $y$ in $G-
\{f_1,f_2\}$. Otherwise, we just continue the search. Then we use the same method in order to determine the reachability from $y$ to $x$ in $G-\{f_1,f_2\}$. Furthermore, we can improve on this idea a little more: even before starting the BFS, we perform this simple check that we have described, in order to see if $x$ reaches a seed, and then if this seed reaches $y$. If the number of seeds is very small (e.g., $10$), then this initial scanning of the seeds takes a negligible amount of time. What is remarkable, is that even with a single random seed, we observe  that random 2FT-SC queries will be answered most of the time even before the BFS begins. We call our implementation of this idea $\mathtt{sBFS}$. Here, the two measures of efficiency are, first, whether the answer was given by a seed (before starting the BFS), and second, what is the total number of edges that we had to access (in case that none of the seeds could provide immediately the answer). We expect the average number of accessed edges to be much lower than in $\mathtt{biBFS}$, because we use the seeds to speed up the search in the process. If $\mathtt{sBFS}$ uses $k$ seeds, we denote the algorithm as $\mathtt{sBFS}(k)$. 

Observe that the checks at the seeds may provide an inconclusive answer (i.e., if either $f_1$ or $f_2$ is an ancestor of the target vertex on the tree-path starting from the seed). Thus, we may instead initialize a 2FT-SSR data structure on every seed, so that every reachability query provides immediately the real answer. In this case, we can extract all the information that the seeds can provide before the BFS begins. We do this by using the four reachability queries $\mathit{2FTR}_r(x,f_1,f_2)$, $\mathit{2FTR}_r(y,f_1,f_2)$, $\mathit{2FTR}_r^R(x,f_1,f_2)$ and $\mathit{2FTR}_r^R(y,f_1,f_2)$, as we explained in Section~\ref{sec:dectree}. We call our implementation of this idea $\mathtt{ChBFS}$ (or $\mathtt{ChBFS}(k)$, to emphasize the use of $k$ seeds). As in $\mathtt{sBFS}$, the two measures of efficiency here are whether one of the seeds provided the answer, or, if not, what is the number of edges that we had to traverse with the bidirectional BFS in order to get the answer.

We consider the queries that force $\mathtt{ChBFS}$ to perform BFS in order to get the answer, worst-case instances for this algorithm. In order to reduce the possibility of such events, we propose the idea of organizing the seeds on a decomposition tree. This reduces the possibility of the worst-case instances, and it may also provide some extra ``free'' seeds on the intermediary levels of the decomposition tree.
We elaborate on this idea in Section~\ref{sec:ch-tree}.

\section{Empirical analysis}
\label{sec:empirical}

We implemented our algorithms in C$\mathtt{{+}{+}}$, using g$\mathtt{{+}{+}}$ v.7.4.0 with full optimization (flag -O3) to compile the code.\footnote{Our code, together with some sample input instances is available at 
\url{https://github.com/dtsok/2-FT-SC-O}.}
The reported running times were measured on a GNU/Linux machine, with Ubuntu (18.04.5 LTS): a Dell PowerEdge R715 server 64-bit NUMA machine with four AMD Opteron 6376 processors and 
185GB of RAM memory. Each processor has 8 cores sharing a 16MB L3 cache, and each core has a 2MB private L2 cache and 2300MHz speed. In our experiments we did not use any parallelization, and each algorithm ran on a single core. We report CPU times measured with the \texttt{high\_resolution\_clock} function of the standard library \texttt{chrono}, averaged over ten different runs.

The real-world graphs that we used in our experiments are reported in Table~\ref{tab:datasets}. From each original graph we extracted its largest SCC, except for $\mathtt{Google\_small}$ for which we used the second-largest SCC, since the largest SCC had more than 400K vertices and 3M edges, and therefore it was too big for our experiments.


\begin{table*}[h!]
	\centering
	\begin{small}
	\begin{tabular}{llrrrrrr}
		\hline
		\multicolumn{1}{l}{Graph} & \multicolumn{1}{l}{Type}  & $n$ & \multicolumn{1}{r}{$m$} &
		\multicolumn{1}{r}{$n_a$} & \multicolumn{1}{r}{$n_{sp}$} & \multicolumn{1}{r}{$d$} & \multicolumn{1}{r}{Reference} \\
\hline
Google\_small & web graph & 950 & 1,969 & 179 & 182 & 10 &\cite{netRep} \\		
Twitter & communication network  & 1,726 & 6,910 & 615 & 1,005 & 18 &\cite{netRep} \\
Rome & road network & 3,353 & 8,870 & 789 & 1,978 & 57 &\cite{ch9-url} \\
Gnutella25 & p2p network & 5,152 & 17,691 & 1,840 & 3,578 & 21 &\cite{snapnets} \\
Lastfm-Asia & social network & 7,624 & 55,612 & 1,338 & 2,455 & 15 &\cite{snapnets} \\
Epinions1 & social network & 32,220 & 442,768 & 8,194 & 11,460 & 16 &\cite{snapnets} \\
NotreDame & web graph & 48,715 & 267,647 & 9,026 & 15,389 & 96 &\cite{snapnets} \\
Stanford & web graph & 150,475 & 1,576,157 & 20,244 & 56,404 & 210 &\cite{snapnets} \\
Amazon0302 & co-purchase graph & 241,761 &	1,131,217 & 69,616 & 131,120 & 88 &\cite{snapnets} \\
USA-road-NY & road network & 264,346 & 733,846 & 46,476 & 120,823 & 720 &\cite{ch9-url} \\
\hline
	\end{tabular}
	\caption{Graph instances used in the experiments, taken from Network Data Repository~\cite{netRep}, 9th DIMACS Implementation Challenge~\cite{ch9-url}, and Stanford Large Network Dataset Collection~\cite{snapnets}.
		$n$ and $m$ are the numbers of vertices and edges, respectively, $n_a$ is the number of strong articulation points (SAPs), $n_{sp}$ is the number of vertices that are SAPs or belong to a proper separation pair, and $d$ is the diameter of the graph. \label{tab:datasets}}
\end{small}
\end{table*}

From Table~\ref{tab:datasets} we observe that a significant fraction of the vertices belong to at least one proper separation pair (value $n_{sp}$ in the table). Indeed, at least $19\%$ of the vertices in every graph, and $44\%$ on average, belong to a proper separation pair.

\subsection{Height of the decomposition tree}
\label{sec:experimentalheight}

We consider various methods for constructing a decomposition tree $T$ of $G$ with small height $h$ in practice.
We note that such decomposition trees are useful in various decremental connectivity algorithms (see, e.g., \cite{Chechik2016Decremental,decdom17,scc-decomposition}), so this experimental study may be of independent interest.
We consider only fast methods for selecting split vertices, detailed in Table~\ref{tab:spilt-algorithms}.
Note that all methods require $O(m)$ time to select a split node $x$ of $G$, except \textsf{Random} and \textsf{LNT}.
\textsf{Random} selects a vertex in constant time, but still requires $O(m+n)$ time to compute the SCCs of $G-x$.
In \textsf{LNT} we use the observation that a loop nesting tree of $G$~\cite{st:t} is a valid decomposition tree, and can be computed in $O(m)$ total time~\cite{dominators:bgkrtw}.

The \textsf{$q$-Separator ($q$Sep)} method is based on the following definition:

\begin{definition} ($q$-separator~\cite{Chechik2016Decremental})
Let $G = (V,E)$ be a graph with $n$ vertices, and let $q \ge 1$ be an integer. A $q$-separator for $G$ is a non-empty set of vertices $S \subseteq V$, such that each SCC of $G \setminus S$ contains at most $n - q \cdot |S|$ vertices.
\end{definition}

Chechik et al.~\cite{Chechik2016Decremental} showed that every strongly connected graph $G$ with $n$ vertices, $m$ edges, and diameter $\delta \ge \sqrt{n}$, has a $q$-separator with quality $q=\sqrt{n}/(2 \log{n})$ that can be computed in $O(m)$ time.
%
%
Here we do not wish to compute the exact diameter of the graph, as this is expensive~\cite{GraphDiameter}. Instead, we apply
\textsf{$q$Sep} if for an arbitrarily chosen start vertex, the longest BFS path in either $G$ or $G^R$ is at least $\sqrt{n}$.
(This folklore $O(m)$-time algorithm gives a $2$-approximation of the diameter~\cite{GraphDiameter}.)
If this is the case, then we remove the $|S|$ vertices of the $q$-separator one at a time.
Otherwise, we need to choose a split vertex by applying some other method. In our experiments, we combined \textsf{$q$Sep} with the \textsf{MostCriticalNode (MCN)} algorithm from \cite{GIP20:SICOMP,Paudel:2018}.

\begin{table*}[h!]
	\tabulinesep=2mm	
	\begin{small}
		\extrarowsep=4pt
		\hspace{-0.4cm}
		\begin{tabular}{p{3cm}p{8.5cm}p{2cm}p{2cm}} 
			\hline
			Algorithm & Technique & Complexity & Reference \\
			\hline

\textsf{Random} & Choose the split vertex uniformly at random  & $O(1)$ &  \\
\textsf{Loop nesting tree (LNT)} & Use a loop nesting tree as the decomposition tree & $O(m)^{\ast}$ & \cite{dominators:bgkrtw,st:t} \\
\textsf{LabelPropagation (LP)} & Partition vertex set into communities and select the vertex with maximum number of neighbors in other communities & $O(m)^{\ast\ast}$ & \cite{Boldi2013,LabelPropagation} \\
\textsf{PageRank (PR)} & Compute the Page Rank of all vertices and return the one with maximum value & $O(m)^{\ast\ast}$ & \cite{PageRank:TR} \\
\textsf{MostCriticalNode (MCN)} & Return the vertex whose deletion minimizes the number of strongly connected pairs & $O(m)$ & \cite{GIP20:SICOMP,Paudel:2018} \\
\textsf{$q$-Separator ($q$Sep)} & Compute a high-quality separator for a graph with a high diameter ($\ge \sqrt{n}$) & $O(m) $ & \cite{Chechik2016Decremental} \\
\textsf{$q$-Separator and MostCriticalNode ($q$Sep$+$MCN)} & If the graph has high diameter then compute a high-quality separator, otherwise compute the MCN & $O(m) $ & \cite{Chechik2016Decremental,GIP20:SICOMP,Paudel:2018} \\
			\hline
		\end{tabular}
		\caption{An overview of the algorithms considered for selecting split vertices of the decomposition tree. The bounds refer to a digraph with $n$ vertices and $m$ edges.
$^{(\ast)}$The stated bound for \textsf{LNT} corresponds to the total time for computing the complete decomposition tree. $^{(\ast\ast)}$The stated bounds for \textsf{LabelPropagation} and \textsf{PageRank} assume that they run for a constant number of iterations.}
		\label{tab:spilt-algorithms}
	\end{small}
\end{table*}

The experimental results for the graphs of Table~\ref{tab:datasets} are presented in Table~\ref{tab:height}, and are plotted in Figure~\ref{figure:dectree-height}. 
We observe that \textsf{MCN} and \textsf{$q$Sep$+$MCN} achieved overall significantly smaller decomposition height compared to the other methods.
In fact, \textsf{Random}, \textsf{LP}, and \textsf{PR} did not manage to produce the decomposition tree for the largest graphs (Amazon0302 and USA-road-NY) in our collection, due to memory or running time restrictions.
For the remaining (smaller) graphs in our collection, \textsf{Random} performed very poorly, giving a decomposition height that was larger by a factor of $11.9$ on average compared to \textsf{MCN}.
%
Also, 
on average, \textsf{LP} and \textsf{LNT} performed better than \textsf{PR}; we see that
\textsf{MCN} produced a tree height that, on average, was smaller by a factor of $2.1$ and $2.5$ compared to \textsf{LP} and \textsf{LNT}, respectively, and by a factor of $4.1$ compared to \textsf{PR}.
Finally, \textsf{MCN} and \textsf{$q$Sep$+$MCN} have similar performance in most graphs, but for the two road networks (Rome and USA-road-NY), \textsf{$q$Sep$+$MCN} produce a significantly smaller decomposition height.

\begin{figure*}[h!]
\begin{center}
\centerline{\includegraphics[trim={0 0 0 2cm}, clip=true, width=\textwidth]{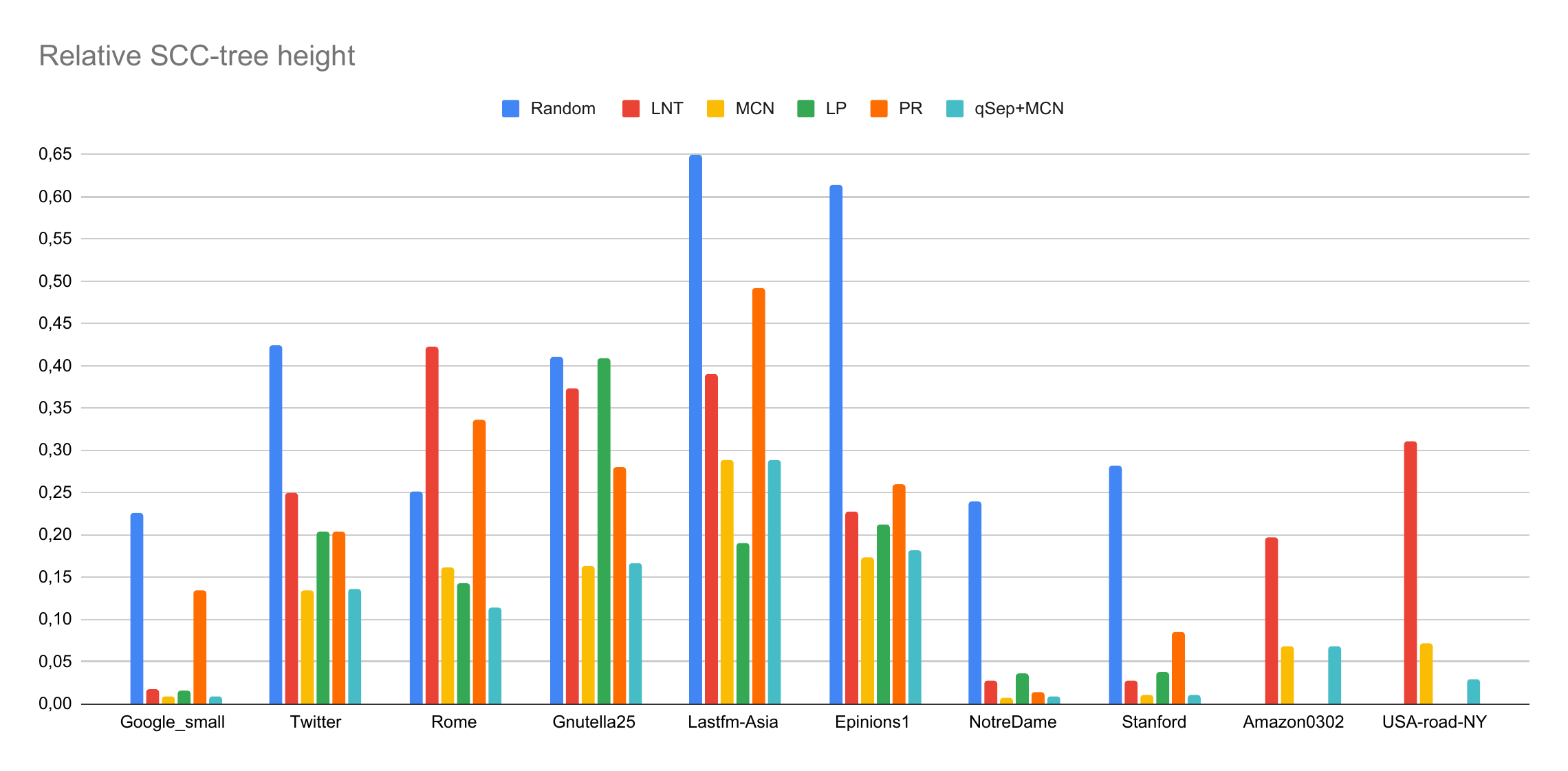}}
\caption{Relative SCC-tree height of the graphs of Table~\ref{tab:datasets} w.r.t. to the number of vertices, resulting from the split vertex selection algorithms of Table~\ref{tab:spilt-algorithms}.\label{figure:dectree-height}}
\end{center}
\end{figure*}

\begin{table*}[h!]
	\centering
 \smaller
	\begin{tabular}{lrrrrrr}

\hline

Graph	& \textsf{Random} & \textsf{LNT} & \textsf{MCN} & \textsf{LP} & \textsf{PR} & \textsf{$q$Sep$+$MCN} \\

\hline

Google\_small & 214 & 16 &  9 & 14 & 128 & 9 \\
Twitter & 732 & 430 & 232 & 352 & 351 & 235 \\
Rome & 843 &1,417 & 542 & 478 & 1,125 & 380 \\
Gnutella25 & 2,118 & 1,920 & 838 & 2,105 & 1,440 & 858 \\
Lastfm-Asia	& 4,958 & 2,981 &2,202 & 1,443 & 3,752 & 2,198 \\
Epinions1 & 19,764 & 7,317 & 5,602 & 6,826 & 8,367 & 5,857 \\
NotreDame & 11,688 & 1,346 & 365 & 1,704 & 672 & 390 \\
Stanford & 42,383 & 4,161 & 1,617 & 5,738 & 12,694 & 1,638 \\
Amazon0302 & $\dag$ & 47,484 & 16,615 & $\ddag$ & $\ddag$ & 16,606 \\
USA-road-NY & $\dag$ & 82,098 & 19,073 & $\ddag$ & $\ddag$ & 7,829 \\
\hline
	\end{tabular}
	\caption{SCC-tree height of the graphs of Table~\ref{tab:datasets}, resulting from the split vertex selection algorithms of Table~\ref{tab:spilt-algorithms}. The symbols $\dag$ and $\ddag$ refer to decompositions that were not completed due to exceeding the RAM memory of our system 
 ($>185$GB) or due to requiring more than $48$ hours.
\label{tab:height}}
\end{table*}

\begin{table}[h!]\centering
\caption{Characteristics of partial SCC-trees achieved by the algorithm of Section~\ref{sec:improved}. 
}\label{tab:ImprovedTree}
\begin{tabular}{lrrrrr}\toprule
Graph & $m$ & $\lfloor \sqrt{m} \rfloor$ & $\Delta$ & height \\\midrule
Google\_small &1,969 &44 &5 &3 \\
Twitter &6,910 &83 &35 &0 \\
Rome99 &8,870 &94 &46 &0 \\
Gnutella25 &17,691 &133 &11 &0 \\
Lastfm-Asia &55,612 &235 &21 &0 \\
NotreDame &267,647 &517 &170 &168 \\
Epinions1 &442,768 &665 &60 &0 \\
Stanford &1,576,157 &1255 &$\leq 700$ &0 \\
\bottomrule
\end{tabular}
\end{table}


Table~\ref{tab:ImprovedTree} reports the characteristics of partial SCC-trees achieved by the algorithm of Section~\ref{sec:improved} for some input graphs. Specifically, we report the height of the partial SCC-tree and the value of the parameter $\Delta$ which gives the maximum number of edges that need to be explored in order to answer a query.
We were not able to include results for the larger graphs in our collection due to high running times to compute the partial SCC-tree. Also, since the computation for Stanford was taking too long, we checked the values $\Delta=100,250,500,700$, and we found that it is $700$-good. Thus, the optimal $\Delta$ for Stanford is $500<\Delta\leq 700$.
We observe that the results are very encouraging. 
For example, we see that Epinions1 is a $60$-good graph. Thus, we can answer every 2FT-SC query on Epinions1 after scanning at most $4\times 60+4$ edges. (We note that Epinions1 has $442$,$768$ edges.) In NotreDame, we may have to reach a tree-node of depth 168 in order to arrive at a $170$-good subgraph, in which we can answer the 2FT-SC query after scanning at most $4\times 170+4$ edges. Before that, we will have to perform $4\times 168$ 2FT-SSR queries on the tree-nodes that we traverse. This is still much faster than the scanning of $\sim{15000}$ edges that we have to perform on average with bidirectional BFS on NotreDame, as reported in Table~\ref{tab:BFS-times}.

\subsection{Answering queries}
\label{sec:experimentalqueries}

First, we consider random queries. We create each query $\twoftsc{x}{y}{f_1}{f_2}$ by selecting the query vertices $x,y$ and the failed vertices $f_1,f_2$ uniformly at random.
The corresponding results for the (basic) SCC-tree are reported in Table~\ref{tab:RandomQueriesTree}.
Evidently, the SCC-tree is very effective, as almost all queries are answered at the root node of the tree.

\begin{table*}[h!]
\small
\hspace{-1cm}
\begin{tabular}{l|c|ccc|cc|cc|cc} \hline 
\multirow{2}{*}{Graph} &tree &\multicolumn{3}{|c}{query depth} &\multicolumn{2}{|c}{query time} &\multicolumn{2}{|c}{query result } &\multicolumn{2}{|c}{avg. calls } \\ 
&depth &min &max &avg. &total (s) &avg. (s) &$+$ &$-$ &$2$-FT-SSR-O &$1$-FT-SC-O \\ \hline 
Google\_small &9 &0 &5 &0.0030 &5.00e-2 &5.00e-8 &987,220 &12,780 &3.976 &0.000181 \\
Twitter &232 &0 &3 &0.0010 &6.18e-2 &6.18e-8 &998,083 &1,917 &3.990 &0.001000 \\
Rome99 &542 &0 &1 &0.0005 &8.70e-2 &8.70e-8 &999,623 &377 &3.997 &0.000575 \\
Gnutella25 &838 &0 &1 &0.0004 &6.70e-2 &6.70e-8 &999,501 &499 &3.998 &0.000410 \\
Lastfm-Asia & 2,202 & 0 & 1 & 0.0002 & 5.40e-2 & 5.40e-8 & 999,851 & 149 & 3.999 & 0.000243 \\
NotreDame & 365 &0 &9 &0.0001 &7.24e-2 &7.24e-8 &999,760 &240 &4.000 &0.000034 \\
Stanford &1,617 &0 &1 &0.0000 &8.99e-2 &8.99e-8 &999,953 &47 &4.000 &0.000016 \\
Epinions1 & 5,602 &0 &1 &0.0000 &6.60e-2 &6.60e-8 &999,920 &80 &4.000 &0.000047 \\
Amazon0302 &16,615 &0 &1 &0.0000 &9.90e-2 &9.90e-8 &999,979 &21 &4.000 &0.000008 \\
USA-NY & 19,073 &0 &1 &0.0000 &3.23e-1 &3.23e-7 &999,993 &7 &4.000 &0.000008 \\
\hline
\end{tabular}
\caption{Results for $1$M random queries using the SCC-tree with split vertices selected by MCN.}\label{tab:RandomQueriesTree}
\end{table*}

In Table~\ref{tab:BFS-times} we see the times for answering 1M queries using $\mathtt{simpleBFS}$ and $\mathtt{biBFS}$. We also report the average number of edge accesses per query, because this is a machine-independent measure of efficiency, and it can also serve as an accurate indicator of the running time of those algorithms. $\mathtt{biBFS}$ charges every edge access a little higher, because every new edge discovery is succeeded by an alteration to the direction of the BFS. From Table~\ref{tab:BFS-times} we can see that $\mathtt{biBFS}$ performs much better than $\mathtt{simpleBFS}$, and thus we use $\mathtt{biBFS}$ as the baseline, and as the last resort when all other heuristics fail to provide the answer.

\begin{table*}[h!]
\hspace{-0.8cm}
\begin{tabular}{l|rrr|rrr}
\hline
\multirow{2}{*}{Graph} &\multicolumn{3}{c|}{simpleBFS } &\multicolumn{3}{c}{biBFS} \\
& \#edges/query &time (s) & edge access (ns) & \#edges/query &time (s) &edge access (ns) \\
\hline
Google\_small &1,527.15 &3.52 &2.30 &148.65 &0.80 &5.40 \\
Twitter &4,777.29 &24.88 &5.21 &179.02 &1.13 &6.31 \\
Rome99 &8,443.75 &56.71 &6.72 &3,665.99 &36.56 &9.97 \\
Gnutella25 &9,749.95 &66.58 &6.83 &302.92 &1.82 &6.01 \\
Lastfm-Asia &40,162.93 &142.29 &3.54 &997.71 &6.03 &6.04 \\
NotreDame &220,197.49 &590.06 &2.68 &15,003.47 &59.36 &3.96 \\
Epinions1 &319,384.06 &660.44 &2.07 &273.72 &1.10 &4.03 \\
Stanford &1,273,970.76 &3,160.50 &2.48 &89,705.73 &449.72 &5.01 \\
Amazon0302 &987,822.00 &8,438.32 &8.54 &19,412.9 &233.06 &12.00 \\
USA-road-NY &742,513.00 &6,493.97 &8.75 &431,299.00 &4,266.94 &9.90 \\
\hline
\end{tabular}
\caption{Results for 1M random queries using $\mathtt{simpleBFS}$ and $\mathtt{biBFS}$. Here is shown the number of edges that we had to access on average per query, as well as the total time for answering all queries on every graph. The third column for every algorithm shows the time in nanoseconds that is charged to every edge access. 
}\label{tab:BFS-times}
\end{table*}


\begin{table*}[h!]
\scriptsize
\hspace{-0.45cm}
\begin{tabular}{lrrrrrrrrr}\toprule
100M CH QUERIES &Google\_small &Twitter &Rome99 &Gnutella25 &Lastfm-Asia &NotreDame &Stanford &Epinions1 \\\midrule
time (s) &1.170 &1.412 &1.988 &1.535 &2.277 &1.477 &1.766 &1.532 \\
time/query (ns) &11.699 &14.116 &19.883 &15.345 &22.768 &14.770 &17.663 &15.325 \\
\bottomrule
\end{tabular}
\caption{\small The total time for answering 100M 2FT-SSR queries using our implementation of Choudhary's $2$FT-SSR data structure \cite{choudhary:ICALP16}. By comparing the times/query with the times per edge access in Table~\ref{tab:BFS-times}, we can see that the time per 2FT-SSR query is comparable to a few edge accesses. We report the average over 10 different random choices of CH-seeds (we note that the variance per graph is negligent). }\label{tab:Choudhary-times}
\end{table*}

Our experiments demonstrate the superiority of $\mathtt{sBFS}$ and $\mathtt{ChBFS}$, compared to $\mathtt{simpleBFS}$ and $\mathtt{biBFS}$. We tested those algorithms with a few number $k$ of seeds, $k\in\{1,2,5,10\}$. In Tables~\ref{tab:Rome99-bfs} to \ref{tab:Stanford-bfs} we report the two indicators of efficiency of the seeded BFS algorithms. That is, we compute the percentage of the queries that are answered simply by querying the seeds (let us call these ``good'' instances), and the average number of edges explored per query, when we have to resort to BFS (in the ``bad'' instances).\footnote{To clarify, these tables do not show the average number of edges explored per bad instance, but the average number of edges explored for all instances. Thus, the average number of edges reported is a good indicator of the relative running times.}

\begin{table}[h!]\centering
\caption{Relative performance of the BFS-based algorithms on Rome99.}\label{tab:Rome99-bfs}
\scriptsize
\begin{tabular}{lrrr}\toprule
\multicolumn{3}{c}{Rome99} \\\cmidrule{1-3}
Algorithm & avg \#edges explored & \% of answer by seed \\\midrule
simpleBFS &8,438.23 & \\
biBFS &3,672.39 & \\
sBFS(1) &100.81 &96.09 \\
sBFS(2) &11.25 &99.47 \\
sBFS(5) &1.87 &99.86 \\
sBFS(10) &0.61 &99.92 \\
ChBFS(1) &2.66 &99.93 \\
ChBFS(2) &0.00 &100.00 \\
ChBFS(5) &0.00 &100.00 \\
ChBFS(10) &0.00 &100.00 \\
\bottomrule
\end{tabular}
\end{table}

\begin{table}[!htp]\centering
\caption{Relative performance of the BFS-based algorithms on Google\_small.}\label{tab:Googlesmall-bfs}
\scriptsize
\begin{tabular}{lrrr}\toprule
\multicolumn{3}{c}{Google\_small} \\\cmidrule{1-3}
Algorithm & avg \#edges explored &\% of answer by seed \\\midrule
simpleBFS &1,527.02 & \\
biBFS &148.72 & \\
sBFS(1) &3.95 &97.98 \\
sBFS(2) &3.11 &98.58 \\
sBFS(5) &3.02 &98.66 \\
sBFS(10) &2.98 &98.69 \\
ChBFS(1) &1.67 &99.07 \\
ChBFS(2) &0.56 &99.80 \\
ChBFS(5) &0.05 &99.95 \\
ChBFS(10) &0.02 &99.97 \\
\bottomrule
\end{tabular}
\end{table}

\begin{table}[!h]\centering
\caption{Relative performance of the BFS-based algorithms on Twitter.}\label{tab:Twitter-bfs}
\scriptsize
\begin{tabular}{lrrr}\toprule
\multicolumn{3}{c}{Twitter} \\\cmidrule{1-3}
Algorithm &avg \# edges explored &\% of answer by seed \\\midrule
simpleBFS &4,777.38 & \\
biBFS &179.05 & \\
sBFS(1) &2.39 &97.89 \\
sBFS(2) &0.48 &99.39 \\
sBFS(5) &0.17 &99.68 \\
sBFS(10) &0.10 &99.74 \\
ChBFS(1) &0.25 &99.86 \\
ChBFS(2) &0.00 &100.00 \\
ChBFS(5) &0.00 &100.00 \\
ChBFS(10) &0.00 &100.00 \\
\bottomrule
\end{tabular}
\end{table}

\begin{table}[!htp]\centering
\caption{Relative performance of the BFS-based algorithms on Gnutella25.}\label{tab:Gnutella25-bfs}
\scriptsize
\begin{tabular}{lrrr}\toprule
\multicolumn{3}{c}{Gnutella25} \\\cmidrule{1-3}
Algorithm &avg \# edges explored &\% of answer by seed \\\midrule
simpleBFS &9,747.45 & \\
biBFS &302.88 & \\
sBFS(1) &1.52 &99.07 \\
sBFS(2) &0.18 &99.84 \\
sBFS(5) &0.04 &99.93 \\
sBFS(10) &0.02 &99.94 \\
ChBFS(1) &0.24 &99.93 \\
ChBFS(2) &0.00 &100.00 \\
ChBFS(5) &0.00 &100.00 \\
ChBFS(10) &0.00 &100.00 \\
\bottomrule
\end{tabular}
\end{table}

\begin{table}[!htp]\centering
\caption{Relative performance of the BFS-based algorithms on Lasfm-Asia.}\label{tab:LasfmAsia-bfs}
\scriptsize
\begin{tabular}{lrrr}\toprule
\multicolumn{3}{c}{Lasfm-Asia} \\\cmidrule{1-3}
Algorithm &avg \# edges explored &\% of answer by seed \\\midrule
simpleBFS &40,131.40 & \\
biBFS &997.50 & \\
sBFS(1) &2.09 &99.75 \\
sBFS(2) &0.29 &99.95 \\
sBFS(5) &0.12 &99.97 \\
sBFS(10) &0.06 &99.98 \\
ChBFS(1) &0.26 &99.97 \\
ChBFS(2) &0.00 &100.00 \\
ChBFS(5) &0.00 &100.00 \\
ChBFS(10) &0.00 &100.00 \\
\bottomrule
\end{tabular}
\end{table}

\begin{table}[!htp]\centering
\caption{Relative performance of the BFS-based algorithms on Epinions1.}\label{tab:Epinions1-bfs}
\scriptsize
\begin{tabular}{lrrr}\toprule
\multicolumn{3}{c}{Epinions1} \\\cmidrule{1-3}
Algorithm &avg \# edges explored &\% of answer by seed \\\midrule
simpleBFS &320,355.93 & \\
biBFS &273.78 & \\
sBFS(1) &0.13 &99.93 \\
sBFS(2) &0.02 &99.98 \\
sBFS(5) &0.01 &99.99 \\
sBFS(10) &0.01 &99.99 \\
ChBFS(1) &0.03 &99.99 \\
ChBFS(2) &0.00 &100.00 \\
ChBFS(5) &0.00 &100.00 \\
ChBFS(10) &0.00 &100.00 \\
\bottomrule
\end{tabular}
\end{table}

\begin{table}[!htp]\centering
\caption{Relative performance of the BFS-based algorithms on NotreDame.}\label{tab:NotreDame-bfs}
\scriptsize
\begin{tabular}{lrrr}\toprule
\multicolumn{3}{c}{NotreDame} \\\cmidrule{1-3}
Algorithm &avg \# edges explored &\% of answer by seed \\\midrule
simpleBFS &220,350.71 & \\
biBFS &14,929.37 & \\
sBFS(1) &13.13 &99.87 \\
sBFS(2) &6.16 &99.94 \\
sBFS(5) &4.85 &99.95 \\
sBFS(10) &3.71 &99.96 \\
ChBFS(1) &2.17 &99.99 \\
ChBFS(2) &0.00 &100.00 \\
ChBFS(5) &0.00 &100.00 \\
ChBFS(10) &0.00 &100.00 \\
\bottomrule
\end{tabular}
\end{table}

\begin{table}[!htp]\centering
\caption{Relative performance of the BFS-based algorithms on Stanford.}\label{tab:Stanford-bfs}
\scriptsize
\begin{tabular}{lrrr}\toprule
\multicolumn{3}{c}{Stanford} \\\cmidrule{1-3}
Algorithm &avg \# edges explored &\% of answer by seed \\\midrule
simpleBFS &1,267,171.57 & \\
biBFS &89,474.38 & \\
sBFS(1) &28.14 &99.96 \\
sBFS(2) &13.67 &99.98 \\
sBFS(5) &10.87 &99.98 \\
sBFS(10) &9.31 &99.98 \\
ChBFS(1) &3.67 &100.00 \\
ChBFS(2) &0.01 &100.00 \\
\bottomrule
\end{tabular}
\end{table}

What is remarkable, is that most queries are answered by simply querying the seeds, and very rarely do we have to resort to BFS. Observe that the higher the number of seeds, the higher the probability that they will provide the answer. However, even with a single seed we get very good chances of obtaining the answer directly. As expected, the seeds in which we have initialized a 2FT-SSR oracle (CH-seeds) have better chances to provide the answer. (In some cases, we get $100\%$ of the answers from the CH seeds.) We note that these results essentially explain the very good times that we observe in Table~\ref{tab:RandomQueriesTree}, since even a single seed can provide the answer to most queries (and thus, only rarely do we have to descend to deeper levels of the decomposition tree).

As we can see in Table~\ref{tab:Choudhary-times}, performing the 2FT-SSR queries on the CH-seeds is a very affordable operation, comparable to accessing a few edges. However, the initialization of $\mathtt{ChBFS}$ takes a lot of time to be completed on larger graphs. (We refer to Tables~\ref{tab:init_times} and \ref{tab:memory} for the initialization time and the memory usage of the BFS-based oracles.) For example, as we can see in Table~\ref{tab:init_times}, $\mathtt{ChBFS}(1)$ takes $6$ seconds to be initialized on Gnutella25, but almost five hours on Amazon0302. On the other hand, the initialization of $\mathtt{sBFS}(10)$ (that uses $10$ seeds), takes less than a second on every graph.  Thus, we can see that $\mathtt{sBFS}$ is an affordable heuristic that works remarkably well in practice.

\begin{table*}[h!]
\scriptsize
\caption{Initialization times for the BFS-based oracles (in seconds).
}
\label{tab:init_times}
\hspace{-0.43cm}
\begin{tabular}{lrrrrrrrrrr}\toprule
Graph &$\mathtt{simpleBFS}$ &$\mathtt{biBFS}$ &$\mathtt{sBFS(1)}$ &$\mathtt{sBFS(2)}$ &$\mathtt{sBFS(5)}$ &$\mathtt{sBFS(10)}$ &$\mathtt{ChBFS(1)}$ &$\mathtt{ChBFS(2)}$ &$\mathtt{ChBFS(5)}$ &$\mathtt{ChBFS(10)}$\\\midrule
Google\_small &2e-06 &1e-05 &4e-05 &7e-05 &2e-04 &4e-04 &2e-05 &4e-05 &1e-04 &2e-04 \\
Twitter &2e-06 &1e-05 &3e-04 &7e-04 &2e-03 &3e-03 &5e-01 &9e-01 &2 &5 \\
Rome99 &2e-06 &1e-05 &3e-04 &7e-04 &2e-03 &3e-03 &5e-01 &8e-01 &2 &4\\
Gnutella25 &2e-06 &1e-04 &6e-04 &1e-03 &3e-03 &6e-03 &6 &11 &28 &56 \\
Lastfm-Asia &2e-06 &1e-04 &1e-03 &3e-03 &7e-03 &1e-02 &12 &23 &60 &120 \\
NotreDame &2e-06 &1e-03 &4e-03 &8e-03 &2e-02 &4e-02 &96 &210 &527 &1054 \\
Epinions1 &2e-06 &3e-03 &8e-03 &2e-02 &4e-02 &8e-02 &575 &1144 &2867 &5734 \\
Stanford &2e-06 &7e-03 &4e-02 &7e-02 &2e-01 &4e-01 &4766 &9142 &23050 &46100 \\
Amazon0302 &2e-06 &9e-03 &5e-02 &1e-01 &3e-01 &5e-01 &16685 &33058 &82800 &165601 \\
USA-road-NY &2e-06 &5e-03 &3e-02 &6e-02 &1e-01 &3e-01 &175 &350 &875 &1750 \\
\bottomrule
\end{tabular}
\end{table*}

\begin{table*}[h!]
\scriptsize
\caption{Memory usage of the BFS-based oracles (in MB).
}
\label{tab:memory}
\hspace{-0.43cm}
\begin{tabular}{lrrrrrrrrrr}\toprule
Graph &$\mathtt{simpleBFS}$ &$\mathtt{biBFS}$ &$\mathtt{sBFS(1)}$ &$\mathtt{sBFS(2)}$ &$\mathtt{sBFS(5)}$ &$\mathtt{sBFS(10)}$ &$\mathtt{ChBFS(1)}$ &$\mathtt{ChBFS(2)}$ &$\mathtt{ChBFS(5)}$ &$\mathtt{ChBFS(10)}$\\\midrule
Google\_small &1.24 &1.28 &1.33 &3.06 &3.06 &3.06 &4.70 &6.70 &11.50 &23.00 \\
Twitter &2.88 &2.98 &2.98 &3.13 &3.13 &3.19 &8.45 &13.65 &28.95 &57.90 \\
Rome99 &3.04 &3.08 &3.17 &3.32 &3.32 &3.58 &16.00 &30.00 &68.00 &136.00 \\
Gnutella25 &3.07 &3.10 &3.15 &3.47 &3.73 &4.13 &21.55 &39.36 &91.89 &183.78 \\
Lastfm-Asia &3.54 &3.39 &3.69 &4.01 &4.28 &5.07 &34.59 &64.81 &155.13 &310.26 \\
NotreDame &5.76 &7.13 &8.60 &10.30 &12.42 &16.38 &143.71 &279.53 &666.55 &1333.10 \\
Epinions1 &6.91 &8.82 &9.68 &10.26 &12.11 &14.49 &106.49 &198.22 &469.88 &939.75 \\
Stanford &17.66 &24.89 &29.33 &31.83 &39.48 &51.63 &517.28 &993.07 &2421.00 &4842.00 \\
Amazon0302 &15.78 &22.23 &29.26 &36.92 &48.27 &67.38 &974.63 &1898.22 &4718.86 &9437.72 \\
USA-road-NY &12.91 &17.92 &25.88 &34.65 &47.06 &67.76 &1515.19 &2987.64 &7450.00 &14900.00 \\
\bottomrule
\end{tabular}
\end{table*}

\subsection{An improved data structure: organizing the CH seeds on a decomposition tree}
\label{sec:ch-tree}

\begin{table*}[ht!]

\scriptsize
\hspace{-0.45cm}
\begin{tabular}{lrrrrr}\toprule
Graph &avg \# edges ChBFS &avg \# edges ChTree &\% of answer by seed in ChBFS &\% of answer by seed in ChTree \\\midrule
Google\_small &145.06 &5.88 &0.07 &95.37 \\
Twitter &179.26 &4.62 &0.07 &96.06 \\
Gnutella25 &302.91 &0.86 &0.02 &99.47 \\
Epinions1 &273.93 &1.42 &0.00 &99.09 \\
Lastfm-Asia &1,007.88 &62.81 &0.00 &91.32 \\
NotreDame &14,989.20 &51.06 &0.01 &99.39 \\
Stanford &89,723.23 &19.38 &0.01 &99.96 \\
Amazon0302 &19,462.16 &2.19 &0.00 &99.98 \\
USA-road-NY &430,818.55 &348.82 &0.00 &99.91 \\
\bottomrule
\end{tabular}
\caption{\small Simulation for answering 10K queries with $\mathtt{ChBFS}$ and $\mathtt{ChTree}$ using $10$ seeds that have high chance to give rise to a bad instance. This experiment was repeated for $100$ different selections of seeds. We see that $\mathtt{ChTree}$ can answer at least $90\%$ of those instances without resorting to BFS. }\label{tab:ChTree}
\end{table*}

Although $\mathtt{ChBFS}$ is the most efficient heuristic for answering the queries, it has mainly two drawbacks. First, initializing the 2-FT-SSR data structures on the seeds is very costly, and thus we cannot afford to use a lot of seeds.
And second, as noted in Section~\ref{sec:BFS}, and as can be seen in Tables~\ref{tab:Rome99-bfs} and \ref{tab:Googlesmall-bfs}, there are instances of queries where the seeds cannot provide the answer, and therefore we have to resort to BFS.

We can make a more intelligent use of the CH-seeds by organizing them on a decomposition tree. More precisely, we use the CH-seeds as split vertices in order to produce an SCC-tree. This confers two advantages. (1) We may get some extra ``free'' CH-seeds on the intermediary levels of the decomposition tree. (2) We essentially maintain all the reachability information that can be provided from the seeds, as if we had initialized them on the whole graph.

Let us elaborate on points (1) and (2). First, initializing Choudhary's data structure for a single vertex takes $O(mn)$ time. However, every level of the decomposition tree has $O(m)$ edges in total. Thus, on every level of the tree, we can afford to initialize as many Choudhary's data structures as are the nodes in it -- at the total cost of initializing a single data structure. This explains (1). Furthermore, on every leaf of the decomposition tree we can afford to initialize an $\mathtt{sBFS}(1)$ data structure, as these data structures demand $O(n)$ space and can be constructed in linear time in total. We call the resulting oracle $\mathtt{ChTree}$. (Observe that this is essentially an SCC-tree-based oracle as described in Section~\ref{sec:dectree}, with the difference that we terminate the decomposition early, since we use only 10 split vertices, and at every leaf of the tree we initialize an $\mathtt{sBFS}(1)$ data structure.) The proof for (2) is essentially given by induction on the level of the decomposition tree, and it is a consequence of the way in which we answer the queries using the seeds (see Section~\ref{sec:dectree}).
Thus, although the data structures on the seeds are initialized on subgraphs of the original graph, whenever we have to use them in order to answer a query, they can provide the same reachability information as if we had initialized them on the whole graph.

We have conducted an experiment in order to demonstrate the superiority of this idea against $\mathtt{ChBFS}$. Specifically, we first observe that the problem of the bad instances is caused by separation pairs whose removal leaves all the seeds concentrated into small SCCs, whereas the query vertices lie in larger components that are unreachable from the seeds. In our experiments, we used $10$ CH-seeds. However, from Tables~\ref{tab:Rome99-bfs} and \ref{tab:Googlesmall-bfs}, we can see that it is very rare to get bad instances from $10$ random seeds. Thus, we have to contrive a way to get seeds that have a high chance to give rise to a lot of bad instances. To do this, we compute the SAPs of the graph, and we process some of them randomly. If for a SAP $s$ the total number of vertices in the SCCs of $G-s$, except the largest one, is at least $10$, then we select randomly $10$ seeds from those components. Then we generate 10K random queries, where one of the failed vertices is $s$, and the query vertices lie in SCCs that do not contain seeds.

On the one hand, we use $\mathtt{ChBFS}$ to answer the queries. As expected, the seeds almost always fail to provide the answer, and so $\mathtt{ChBFS}$ can do no better than resort to $\mathtt{biBFS}$. (However, sometimes we manage to squeeze out a negative answer from the seeds, due to Observation~\ref{obervation:sc}.) On the other hand, we use  the SCC-decomposition tree $\mathtt{ChTree}$ to answer the queries.

Since the data structure of Choudhary takes a lot of time to be initialized, we could perform a large number of those experiments only by simulating the process of answering the queries using a 2FT-SSR oracle. That is, at every node of the decomposition tree, we simply used bidirectional BFS in order to determine whether the split vertex that corresponds to it reaches the query vertices. In this way, we can still report exactly the percentage of the queries that can be answered without resorting to BFS at the leaves of the decomposition tree. As we can see in Table~\ref{tab:ChTree}, more than $90\%$ of the generated bad instances can be answered by $\mathtt{ChTree}$ without performing BFS, by using only the data structures on the decomposition tree (that has height at most $10$).

\section{Concluding remarks}
Our experiments demonstrate that $\mathtt{sBFS}$ is a remarkably good heuristic for efficiently answering 
$2$-FT-SC queries in practice.  
By relying on the 2FT-SSR oracle of Choudhary~\cite{choudhary:ICALP16}, we can improve the accuracy of this heuristic, and the organization of the CH-seeds into a decomposition tree minimizes the likelihood of bad instances. It seems that picking the most critical nodes for an SCC-tree decomposition with an early termination at a few levels is the best choice for applications, since these nodes decompose the graph quickly into small SCCs, and thus we increase the likelihood of answering the queries fast, using a few $O(1)$-time calls to auxiliary data structures. Finally, we note that our BFS-based heuristic (seeded BFS) seems to be applicable also in answering $k$-FT-SC queries, for small $k\geq 3$. It is an interesting question what is a good choice of seeds that can increase the efficiency of this heuristic.

\bibliographystyle{alpha}
\bibliography{ltg}

\end{document}